\documentclass[letterpaper, 10 pt, conference]{ieeeconf}  

\IEEEoverridecommandlockouts                              

\usepackage[T1]{fontenc}
\usepackage{graphicx}
\usepackage{soul}
\usepackage{xcolor}
\usepackage{colortbl}  
\usepackage[english]{babel}
\usepackage{amsthm}
\usepackage{amssymb}
\newtheorem{assumption}{\textbf{Assumption}}
\newtheorem{theorem}{\textbf{Theorem}} 

\newtheorem{corollary}[theorem]{\textbf{Corollary}}
\newtheorem{proposition}{\textbf{Proposition}}
\newtheorem{definition}{\textbf{Definition}}
\newtheorem{remark}{\textbf{Remark}}

\usepackage{amsmath,amssymb,amsfonts}
\usepackage{algorithmic}
\usepackage{graphicx}
\usepackage{textcomp}
\usepackage{colortbl}  
\usepackage{xcolor}
\usepackage[colorlinks=true, linkcolor=blue, citecolor=blue, urlcolor=blue, hypertexnames=false]{hyperref}

\usepackage{amsmath,amssymb,amsthm}
\usepackage[T1]{fontenc}

\title{\LARGE \bf
Optimal Modified Feedback Strategies in LQ Games under Control Imperfections
}

\author{Mahdis Rabbani$^{1}$, Navid Mojahed$^{1}$, and Shima Nazari$^{1}$
\thanks{*This work was not supported by any organization.}
\thanks{$^{1}$Mahdis Rabbani, Navid Mojahed, and Shima Nazari are with the Mechanical and Aerospace Engineering,
        University of California, Davis, 1 Shields Ave., Davis, 95616 CA, USA
        {\tt\small \{mrabbani, nmojahed, snazari\}@ucdavis.edu}}%
}

\begin{document}

\addtolength{\topmargin}{2mm}   

\maketitle
\thispagestyle{empty}
\pagestyle{empty}

\begin{abstract}
Game-theoretic approaches and Nash equilibrium have been widely applied across various engineering domains. However, practical challenges such as disturbances, delays, and actuator limitations can hinder the precise execution of Nash equilibrium strategies. This work investigates the impact of such implementation imperfections on game trajectories and players' costs in the context of a two--player finite--horizon linear quadratic (LQ) nonzero-sum game. Specifically, we analyze how small deviations by one player, measured or estimated at each stage affect the state trajectory and the other player’s cost. To mitigate these effects, we construct a compensation law for the influenced player by augmenting the nominal game with the measurable deviation dynamics. The resulting policy is
shown to be optimal within a causal affine policy class, and,
for sufficiently small deviations, it locally outperforms the
uncompensated equilibrium-derived feedback. Rigorous analysis
and proofs are provided, and the effectiveness of the proposed
approach is demonstrated through a representative numerical
example.

\end{abstract}

\section{Introduction}
\label{sec:introduction}

Game theory is a powerful analytical framework for modeling rational decision-making in complex multi-agent systems, particularly when interactions are strategic and agents’ decisions influence each other~\cite{shamma2020game,basar1999dynamic}. Such interactions often involve conflicting objectives, motivating the development of solution concepts to characterize equilibrium behavior. Among these, the Nash equilibrium is a well-established notion, defined as a set of strategies from which no participant can unilaterally deviate to achieve a better outcome~\cite{nash1950equilibrium}.

Thanks to its solid theoretical foundation, Nash equilibrium has been widely applied in diverse engineering domains, including robotics~\cite{chiu2021encoding}, traffic management~\cite{Oszczypała2023Nash}, and sensor networks~\cite{Zheng2023game}, as well as in emerging areas such as energy management~\cite{YANG2020216nashQ,abdolmohammadi2024sizing}, autonomous vehicle navigation~\cite{nazari2024powertrain, bouzidi2025interaction, abtahi2025multi}, and human–robot interaction~\cite{soltanian2025peerawarecostestimationnonlinear}.

Despite its wide adoption, implementing Nash strategies in real-world engineering systems faces significant challenges due to uncertainties, disturbances, delays, and actuator imperfections~\cite{van-den=broek2023,AMATO2002507_guaranteeing,ma2024ϵ}. Extensions such as the $\varepsilon$-equilibrium relax optimality by allowing bounded deviations from Nash strategies~\cite{TANAKA1991413}, enabling numerical approximations~\cite{NORTMANN20231760,alvarez2009urban} and supporting further studies on sensitivity and robustness~\cite{van-den=broek2023,AMATO2002507_guaranteeing,Jimenez01072006,tan2005sensitivity,hernandez2024ϵ}.

Robustness of Nash equilibrium to uncertainty has been studied from various perspectives. For example, model uncertainties in differential games were addressed in~\cite{van-den=broek2023,AMATO2002507_guaranteeing,hernandez2024ϵ}, where~\cite{van-den=broek2023} designed a finite-horizon feedback Nash strategy for Linear Quadratic (LQ) games, \cite{AMATO2002507_guaranteeing} extended the scope to infinite-horizon settings, and \cite{hernandez2024ϵ} relaxed the LQ structure via weaker convexity assumptions. Other works addressed different uncertainty sources:~\cite{Jimenez01072006} proposed adaptive control to mitigate measurement noise, and~\cite{tan2005sensitivity} analyzed the effect of cost parameter variations in discrete-time games.

Nonetheless, the sensitivity of Nash equilibrium to player-specific control deviations, arising from implementation imperfections, remains largely unexplored. Such deviations, e.g., due to actuator errors or delays, are common in engineering systems and can directly degrade performance. Moreover, the discrete-time setting has received little attention in this context. This work analytically investigates how such deviations influence system trajectories and player costs.

To address these effects, we construct a compensation law for
the influenced player that uses the nominal equilibrium structure
together with measurable deviation information. The resulting
controller is not posed as a new Nash equilibrium of the modified
interaction; rather, it is a structured compensation policy around
the nominal game that mitigates performance loss caused by
execution-time deviations.
While prior adaptive multi-agent methods exist~\cite{guerrero2021openloop,wang2022risk}, they do not explicitly target compensation for execution-time control deviations of a specific player, which is addressed in this work.

This paper investigates finite-horizon, discrete-time LQ
games in which one player’s control inputs deviate during execution.
We analyze the resulting impact on joint trajectories and on the
other player’s cost. The main contributions are summarized below:
\begin{itemize}
    \item A sensitivity analysis for discrete-time games, deriving closed-form expressions that characterize how one player’s deviations propagate through the system dynamics and influence the other player’s cost.
    \item
    A compensation design for the influenced player, obtained
  by augmenting the nominal FNE dynamics with measurable deviation dynamics. The resulting policy is shown to be uniquely optimal over a causal affine policy class that includes the uncompensated reference feedback as a special case.
    \item  A local performance comparison showing that, for sufficiently small deviations, the proposed compensated feedback is no worse than the uncompensated equilibrium-derived feedback to second order, together with a representative numerical example illustrating the resulting performance recovery.
\end{itemize}

The remainder of the paper is organized as follows.
Section~\ref{sec: problem statement} presents the problem formulation and background.
Section~\ref{sec:error-on-nash-equilibria} develops the sensitivity analysis.
Section~\ref{sec: control policy design} introduces the modified feedback control policy.
Section~\ref{sec:numerics} demonstrates its effectiveness through a numerical example.
Finally, Section~\ref{sec: conclusion} concludes the paper.

\noindent \textbf{Notation}. $\mathbb{R}$ and $\mathbb{Z}$ denote the sets of real and integer numbers, respectively. For a matrix $A$, $A^\top$ denotes its transpose; for a square matrix $B$, $B \succ 0$ ($B \succeq 0$) means $B$ is positive definite (positive semi-definite), and $B^{-1}$ denotes its inverse. For $i \in N_p=\{1,2\}$, the subscript $_{-i}$ refers to $N_{p} \setminus \{i\}$, i.e., all indices except $i$. The superscript $^\star$ denotes the exact optimal solution. A tilde, as in $\tilde P_{i,k}$, denotes auxiliary Riccati-type matrices used to distinguish the coupled finite-horizon game recursion from standard Riccati matrices.

\section{Preliminaries}
\label{sec: problem statement}
We focus on finite-horizon linear quadratic (LQ) games~\cite{Engwerda2009lq} with discrete-time dynamics
\begin{equation}
    \label{eq:DT game dynamics}
    x_{k+1} = A x_k + B_1 u_{1,k} + B_2 u_{2,k},
\end{equation}
for $k \in \mathbb{K} = \{0,1,\dots, N-1\} \subset \mathbb{Z}$, where:
\begin{itemize}
    \item \(x_k = \begin{bmatrix} x_{1,k}^\top & x_{2,k}^\top \end{bmatrix}^\top
\in \mathbb{R}^{n_1+n_2}\) is the joint state vector obtained by concatenating the state vectors $x_{i,k} \in \mathbb{R}^{n_i}$ of Player $i \in \{1,2\}$, where $n_i$ is the dimension of Player~$i$'s state.
    \item \(u_{1,k} \in \mathbb{R}^{m_1}\) and \(u_{2,k} \in \mathbb{R}^{m_2}\) are the control inputs of Player~1 and~2, respectively, where $m_i$ is the dimension of Player~$i$'s control input.
    \item \(A \in \mathbb{R}^{(n_1+n_2)\times (n_1+n_2)}\) is the joint state transition matrix.
    \item \(B_1 \in \mathbb{R}^{(n_1+n_2)\times m_1}\) and \(B_2 \in \mathbb{R}^{(n_1+n_2)\times m_2}\) are the input matrices for Player~1 and~2, respectively.
\end{itemize}

Suppose Player~$i \in \{1,2\}$ minimizes its cost $J_i$ over an $N$-step horizon by choosing control inputs $u_{i,k}$ at each stage $k \in \mathbb{K}$. The strategy of Player~$i$ is $U_i = \{u_{i,k}\}_{k=0}^{N-1}$, and $U_{-i}$ denotes its counterpart’s strategy. The quadratic cost function of Player~$i$ is
\begin{equation}
    \label{eq:LQ cost function}
    J_i = \sum_{k=0}^{N-1} \left( x_k^\top Q_i x_k + u_{i,k}^\top R_i u_{i,k} \right) + x_N^\top Q_{i,N} x_N,
\end{equation}
where $Q_i, Q_{i,N} \succeq 0$ are the stage and terminal weighting matrices, and $R_i \succ 0$ is the control weighting matrix.

We use Nash equilibrium to mean a strategy profile $U_i$, where no player can reduce its cost by a unilateral deviation, i.e. $ J_i(x_0, U_i^\star, U_{-i}^\star) \le J_i(x_0, U_i, U_{-i}^\star)$~\cite{nash1950equilibrium}.
\begin{definition}
\label{def: definition of Nash}
Consider the LQ game \eqref{eq:DT game dynamics}–\eqref{eq:LQ cost function}.  
A pair of feedback laws
\[
u_{i,k}=-K^\star_{i,k}x_k,\qquad i\in\{1,2\},\ k\in\mathbb{K},
\]
constitutes a linear Feedback Nash Equilibrium (FNE) if, for each $i$, $u_{i,k}$ minimizes $J_i$ against the opponent’s feedback $u_{-i,k}=-K^\star_{-i,k}x_k$ over the horizon.
\end{definition}

The FNE gains $\{K^\star_{i,k}\}$ are obtained by a backward Riccati recursion coupled across players \cite{salizzoni2025bridging}.  
Let $\tilde P_{i,N}=Q_{i,N}$ and define, for each $k\in\{0,\dots,N-1\}$ (computed backward),
\[
A^{-i}_{\mathrm{cl},k} \;:=\; A-\sum_{j\neq i} B_j K_{j,k}.
\]
Given $\{\tilde P_{i,k+1}\}$ and $\{K_{-i,k}\}$, the matrices $\tilde P_{i,k}$ satisfy
\begin{equation}
\label{eq:coupled_riccati-nominal}
\begin{aligned}
\tilde P_{i,k}
= Q_i &+ (A^{-i}_{\mathrm{cl},k})^\top \tilde P_{i,k+1} A^{-i}_{\mathrm{cl},k}\\
&- (A^{-i}_{\mathrm{cl},k})^\top \tilde P_{i,k+1} B_i
\,\tilde H_k\,^{-1}\!
B_i^\top \tilde P_{i,k+1} A^{-i}_{\mathrm{cl},k},
\end{aligned}
\end{equation}
where $\tilde H_k :=R_i + B_i^\top \tilde P_{i,k+1} B_i$. The corresponding FNE feedback gains are
\begin{equation}
\label{eq:Ki_formula-nominal}
K^\star_{i,k}
= \tilde H_k^{-1} B_i^\top \tilde P_{i,k+1} A^{-i}_{\mathrm{cl},k}.
\end{equation}

Under the standard assumptions $Q_i,Q_{i,N}\succeq 0$, $R_i\succ 0$ and nonsingularity of the gain system  in Definition~\ref{def: definition of Nash} at each $k$, the backward sweep $k=N-1\to 0$ yields a unique FNE \cite[Sec.~6]{Engwerda2009lq}.


We next recall standard sufficient conditions that ensure the existence and uniqueness of the Nash equilibrium in Definition~\ref{def: definition of Nash}, thereby guaranteeing well-posedness of the game.

\begin{remark}
Consider the following sufficient conditions for the existence of a unique Nash equilibrium in a two-player finite-horizon nonzero-sum discrete-time LQ game \eqref{eq:LQ cost function}:
\begin{enumerate}
    \item The state weighting matrices satisfy $Q_{i}, Q_{i,N} \succeq 0$ and the control weighting matrices satisfy $R_{i} \succ 0$ for $i \in \{1,2\}$.\label{condition 1}
    \item Let $\tilde{B} = \begin{bmatrix} B_1 & B_2 \end{bmatrix}$. The pair $(A, \tilde{B})$ is controllable.\label{condition 2}
    \item Define
    \[
    M_{k} =
    \begin{bmatrix}
    I & F_1 & F_2 \\
    P_{1,k} & -I & 0 \\
    P_{2,k} & 0 & -I
    \end{bmatrix},
    \quad F_i = B_i R_{i}^{-1} B_i^\top,
    \]
    where $I$ is the $n\times n$ identity and $P_{i,k}$ denote the player Riccati matrices from the standard backward recursions; invertibility of $M_k$ ensures well-posed coupled updates.\label{condition 3}
\end{enumerate}
Under conditions~\ref{condition 1}-\ref{condition 3}, a unique finite-horizon Nash equilibrium exists. We assume these conditions henceforth \cite{LAA_Jungers_2013}.
\end{remark}

\section{First-Order Sensitivity of Player~1 Cost under Opponent Deviations}
\label{sec:error-on-nash-equilibria}
Under the memoryless state-feedback information structure, a finite-horizon
FNE for a two-player game has the form
\begin{equation}
\label{eq:feedback nash definition}
u^\star_{i,k} = -K^\star_{i,k} x_k, \qquad i\in\{1,2\},\ \ k=0,\dots,N-1,
\end{equation}
as stated in Definition~\ref{def: definition of Nash}.
Substituting \eqref{eq:feedback nash definition} into the dynamics \eqref{eq:DT game dynamics} yields the
time-varying closed-loop matrix
\begin{equation}
\label{eq:Acl-def}
A_{\mathrm{cl},k} := A - B_1 K^\star_{1,k} - B_2 K^\star_{2,k},
\end{equation}
which evolves the dynamics
\begin{equation}
\label{eq:closed-loop Nash dynamic}
x_{k+1}^\star = A_{\mathrm{cl},k} x_k^\star,\qquad x_0^\star \equiv x_0.
\end{equation}

In practice, Player~2’s implemented input may deviate from its FNE command,
$u_{2,k} = u^\star_{2,k} + \Delta u_{2,k}$. In this section, we quantify the
effect of small deviations $\{\Delta u_{2,k}\}$ on Player~1’s cost when
Player~1 keeps its FNE law \eqref{eq:feedback nash definition}. We obtain
(i) an exact expression for the induced state perturbation and
(ii) a first-order expansion of $\Delta J_1$ with explicit coefficients that
respect the time variation in \eqref{eq:Acl-def}.

\begin{definition}
\label{def:phi}
For integers $l\ge k\ge j$, define
\begin{equation*}
\Phi(k,j) := \prod_{i=j}^{k-1} A_{\text{cl},i} \;=\; A_{\text{cl},k-1} \cdots A_{\text{cl},j},
\end{equation*}
with the convention $\Phi(j,j):=I$. 
\end{definition}
The state-transition matrix satisfies the standard identities
\[
\Phi(k+1,j)=A_{\text{cl},k}\,\Phi(k,j), \qquad\Phi(k,\ell)\,\Phi(\ell,j)=\Phi(k,j).
\]

\begin{proposition}\label{prop:first_order}
Let $\{K^\star_{1,k}, K^\star_{2,k}\}_{k=0}^{N-1}$ be the FNE gains in \eqref{eq:feedback nash definition}. Suppose Player~2 deviates as $u_{2,k}=u^\star_{2,k}+\Delta u_{2,k}$ while Player~1 keeps $u_{1,k}=u^\star_{1,k}=-K^\star_{1,k}x_k$. Then, for $k=1,\dots,N$,
\begin{equation}\label{eq:dx_formula}
\Delta x_k \;=\; \sum_{j=0}^{k-1} \Phi(k,j+1)\,B_2\,\Delta u_{2,j},
\end{equation}
and, with $S_k := Q_1 + K_{1,k}^{\star\top} R_1 K^\star_{1,k}$,
\begin{equation}\label{eq:DJ_linear}
\Delta J_1 \;=\; 2 \sum_{j=0}^{N-1} x_0^\top \Lambda_j \,B_2\,\Delta u_{2,j} \;+\; \mathcal{O}\!\big(\|\Delta u_2\|_2^2\big),
\end{equation}
where
\begin{equation}\label{eq:Lambda_j}
\begin{aligned}
    \Lambda_j \;=\; \sum_{k=j+1}^{N-1} \Phi(k,0)^\top &S_k \,\Phi(k,j+1) \\
    \;&+\; \Phi(N,0)^\top Q_{1,N}\,\Phi(N,j+1),
\end{aligned}
\end{equation}
for $j=0,\dots,N-1$.
Here $\Phi(k,j)$ is as in Definition~\ref{def:phi}, $\Delta u_2:=[\Delta u_{2,0}^\top\,\cdots\,\Delta u_{2,N-1}^\top]^\top$, and
the remainder is with respect to the Euclidean norm $\|\Delta u_2\|_2$.
\end{proposition}

\begin{proof}
Under the FNE
, the nominal trajectory is obtained via \eqref{eq:closed-loop Nash dynamic}. With $u_{2,k}=u_{2,k}^\star+\Delta u_{2,k}$ and $u_{1,k}=u_{1,k}^\star$, the actual state satisfies
\begin{equation*}
x_{k+1} \;=\; A_{\text{cl},k} \, x_k + B_2 \, \Delta u_{2,k}.
\end{equation*}
Defining $\Delta x_k := x_k - x_k^\star$ gives
\begin{equation*}\label{eq:dx_recursion}
\Delta x_{k+1} \;=\; A_{\text{cl},k} \, \Delta x_k + B_2 \, \Delta u_{2,k}, \qquad \Delta x_0 = 0,
\end{equation*}
whose iteration yields \eqref{eq:dx_formula} using $\Phi(k+1,j)=A_{\text{cl},k}\,\Phi(k,j)$.

Next, linearizing \eqref{eq:LQ cost function} for $i=1$ around the FNE trajectory $(x_k^\star,u_{1,k}^\star)$, leads to
\begin{equation*}\label{eq:DJ_raw}
\begin{aligned}
    \Delta J_1 \;=\; 2 \sum_{k=0}^{N-1} \big( x_k^{\star\top}& Q_1 \Delta x_k + u_{1,k}^{\star\top} R_1 \Delta u_{1,k} \big)\\
    \;+\; &2 \, x_N^{\star\top} Q_{1,N} \Delta x_N \;+\; \mathcal{O}(\|\Delta x\|_2^2),
\end{aligned}
\end{equation*}
where $\Delta\,J_1 = J_1(x_k,u_{1,k})-J_1(x_k^\star,u_{1,k}^\star)$.

Since $u_{1,k}=-K^\star_{1,k}x_k$, we have $\Delta u_{1,k} = -K^\star_{1,k}\Delta x_k$, hence
\begin{equation}\label{eq:DJ_Sk}
\begin{aligned}
    \Delta J_1 \;=\; 2 \sum_{k=0}^{N-1} x_k^{\star\top} S_k \, \Delta x_k
\;+\; 2 x_N^{\star\top}& Q_{1,N} \Delta x_N \\
\;&+\; \mathcal{O}(\|\Delta x\|_2^2),
\end{aligned}
\end{equation}
where $S_k := Q_1 + K_{1,k}^{\star\top} R_1 K^\star_{1,k}$. Substituting \eqref{eq:dx_formula} into \eqref{eq:DJ_Sk}, defining $\Lambda_j$ as in \eqref{eq:Lambda_j}, and using $x_k^\star=\Phi(k,0)x_0$ gives \eqref{eq:Lambda_j}. As $\Delta x$ is linear in $\Delta u_2$ by \eqref{eq:dx_formula}, the remainder is $\mathcal{O}(\|\Delta u_2\|_2^2)$. This completes the proof.
\end{proof}

Proposition~\ref{prop:first_order} shows the deviations of Player~2 linearly perturb Player~1’s cost. If Player~1 keeps the equilibrium-derived feedback without compensation, e.g. $u_{1,k}^{\mathrm{REF}}:=-K^\star_{1,k}x_k$ (referred to as \textit{Reference Feedback (REF)}), its policy is generally sub-optimal under the realized dynamics. Motivated by this sensitivity, we next design a controller that models the deviation and attenuates its effect as it propagates through the closed loop, compensating for these deviations.

\section{Control Policy Design with Disturbance Compensation}
\label{sec: control policy design}

In practice, the control implementation is rarely exact. Differences between the commanded and implemented inputs typically arise from (i) actuator/servo dynamics and inner-loop tracking (first-order lags, rate/slew limits), (ii) computation/communication latency and sample-and-hold effects, and (iii) command filtering/smoothing or safety layers that reshape inputs.
These effects induce an input disturbance $w_k \;:=\; B_2\big(u_{2,k}-u_{2,k}^\star\big)$, acting precisely through the channel identified by the sensitivity result.

In practice, $w_k$ need not be directly measured at the actuator. Instead, Player~1 may form an estimate $\hat w_k$ from available information. If the commanded input $u^\star_{2,k}=-K^\star_{2,k}x_k$ is known from the nominal FNE model and the implemented input $u_{2,k}$ can be communicated, sensed, or reconstructed from actuator or vehicle telemetry, then one may compute $\hat w_k = B_2\bigl(u_{2,k}-u^\star_{2,k}\bigr)$.
When $u_{2,k}$ is not directly available, $\hat w_k$ can instead be obtained using a standard disturbance observer or state estimator built from the plant model and measured state trajectory. In that case, the proposed CF law is applied using $\hat w_k$ in place of $w_k$, while the present analysis corresponds to the ideal full-information case.

We therefore model $w_k$ using simple, identifiable dynamics and equip Player~1 with a \textit{Compensated Feedback (CF)} law containing a feedforward term $L_k w_k$ to offset the predictable component of the disturbance. In this way, the feedback term $K_k$ regulates the residual error rather than reacting only after the state has already deviated. We first consider a physically grounded actuator-lag model and then narrow it down to a slow-varying error model.

\subsection{Compensated Feedback under Lagged Execution Error Dynamics}
\label{sec:lag-best-response}

Assume that Player~2’s actuator is a discrete-time first-order lag that tracks its desired (Nash) input $u^{\star}_{2,k}=-K^{\star}_{2,k}x_k$ and let $u_{2,k}$ denote the implemented command. Let the execution error be $e_k:=u_{2,k}-u^{\star}_{2,k}$. Then, the plant disturbance is $w_k=B_2 e_k$.

\begin{assumption}
\label{ass:info-lag}
At time $k$, Player~1 observes $(x_k,w_k)$ and Player~2 implements its FNE command $u^{\star}_{2,k}=-K^{\star}_{2,k}x_k$ through a discrete first-order lag:
\begin{equation}\label{eq:u2-lag}
u_{2,k+1} \;=\; \alpha\,u_{2,k} \;+\; (1-\alpha)\,u^{\star}_{2,k}, \qquad \alpha\in(0,1).
\end{equation}
The induced disturbance recursion is
\begin{equation}\label{eq:w-lag}
\begin{aligned}
w_{k+1} &\;=\; \alpha\,w_k \;-\; B_2\big(u^{\star}_{2,k+1}-u^{\star}_{2,k}\big)\\
&\;=\; \alpha\,w_k \;+\; B_2\!\big(K^{\star}_{2,k+1}x_{k+1}-K^{\star}_{2,k}x_k\big).
\end{aligned}
\end{equation}
It is noteworthy that at each stage $k$, Player 1 knows the current $x_k$ and the current $w_k$, but not future values.
\end{assumption}

\begin{remark}
\label{rem:alpha-tau}
The discrete recursion \eqref{eq:u2-lag} is the zero-order-hold (ZOH) sampled counterpart of the continuous-time actuator $\dot u_2(t)=\tfrac{1}{\tau}(u_2^{\star}(t)-u_2(t))$ with sampling period $\Delta t$. The pole maps as
\begin{equation}\label{eq:alpha-tau}
\alpha \;=\; e^{-\Delta t/\tau},
\end{equation}
so the time constant $\tau$ of the continuous model sets the per-sample settling factor $\alpha\in(0,1)$.
\end{remark}

The lag model stated in Assumption~\ref{ass:info-lag}  captures two ubiquitous effects: (i) if the target $u^{\star}_{2,k}$ is steady, the tracking error decays geometrically with factor $\alpha$; (ii) changes in $u^{\star}_{2,k}$ inject a transient error that then dies out. Such behavior is widely used as a realistic low-order approximation of implementation effects.

Fix Player~2 at its Nash feedback $u^{\star}_{2,k}=-K^{\star}_{2,k}x_k$ and define
\begin{equation}\label{eq:Aeff-def}
A_{\mathrm{eff},k}:=A-B_2K^{\star}_{2,k},\qquad
x_{k+1}=A_{\mathrm{eff},k}x_k+B_1u_{1,k}+w_k .
\end{equation}

\begin{theorem}
\label{thm:lag-optimal}
Consider the game \eqref{eq:LQ cost function} with dynamics \eqref{eq:DT game dynamics}.
Let Player~2 obey its FNE policy \eqref{eq:feedback nash definition} and define the effective matrices
$A_{\mathrm{eff},k}$ via \eqref{eq:Aeff-def}. Under Assumption~\ref{ass:info-lag}, let the disturbance $w_k$ be generated by the lag recursion
\eqref{eq:w-lag}. Assume $Q_1,Q_{1,N}\succeq 0$ and $R_1\succ 0$. Over the causal affine class
\begin{equation}\label{eq:policy-class}
u_{1,k} \;=\; -K_k x_k \;-\; L_k w_k , \qquad k=0,\dots,N-1,
\end{equation}
there exists a \emph{unique} optimal policy that minimizes \eqref{eq:LQ cost function}, with gains
\begin{equation}\label{eq:gains-closed}
\big[\,K_k\ \ L_k\,\big]
\;=\;
\big(R_1+\bar B_k^\top \bar P_{k+1}\bar B_k\big)^{-1}\,\bar B_k^\top \bar P_{k+1}\bar A_k,
\end{equation}
where $\{\bar P_k\}$ is the unique sequence generated by the finite-horizon Riccati recursion
\begin{equation}\label{eq:riccati-aug}
\begin{aligned}
&\bar P_N=\bar Q_N,\\
&\begin{aligned}
\bar P_k&=\bar Q+\bar A_k^\top \bar P_{k+1}\bar A_k\\
&-\bar A_k^\top \bar P_{k+1}\bar B_k\big(R_1+\bar B_k^\top \bar P_{k+1}\bar B_k\big)^{-1}\bar B_k^\top \bar P_{k+1}\bar A_k,
\end{aligned}\\
&\text{for } k=N-1,\dots,0,
\end{aligned}
\end{equation}
on the augmented system with matrices
\begin{equation}\label{eq:ABQ-def}
\begin{aligned}
&\bar A_k:=\begin{bmatrix}
A_{\mathrm{eff},k} & I\\[2pt]
B_2K^{\star}_{2,k+1}A_{\mathrm{eff},k}-B_2K^{\star}_{2,k} & \alpha I + B_2K^{\star}_{2,k+1}
\end{bmatrix},\\
&\bar B_k:=\begin{bmatrix}B_1\\[2pt] B_2K^{\star}_{2,k+1}B_1\end{bmatrix},
\end{aligned}
\end{equation}
where
\begin{equation}\label{eq:Q-aug}
    \bar{Q} = \begin{bmatrix} Q_1 & 0 \\ 0 & 0 \end{bmatrix}, \quad \bar{Q}_N = \begin{bmatrix} Q_{1,N} & 0 \\ 0 & 0 \end{bmatrix}.
\end{equation}
Moreover, for any causal $\{\tilde K_k,\tilde L_k\}$, the incurred cost satisfies the nonnegative gap identity
\begin{equation}\label{eq:gap}
\begin{aligned}
\sum_{k=0}^{N-1}\!\big(\tilde u_{1,k}-u_{1,k}^\star\big)&^\top
\big(R_1+\bar B_k^\top \bar P_{k+1}\bar B_k\big)
\big(\tilde u_{1,k}-u_{1,k}^\star\big)\\
&\,=\,J_1(\{\tilde K_k,\tilde L_k\})-J_1(\{K_k,L_k\})\ \ge 0,
\end{aligned}
\end{equation}
where $u_{1,k}^\star:=-K_k x_k - L_k w_k$ and $\tilde u_{1,k}:=-\tilde K_k x_k - \tilde L_k w_k$.
\end{theorem}

\begin{proof}
Inspired by the Disturbance-Based State Compensation (DBSC) approach \cite{rugh1996linear}, which is commonly used in traditional state-feedback control design, and given Assumption~\ref{ass:info-lag}, we define the augmented state as $z_k:=[x_k^\top\ w_k^\top]^\top$.

Then the game problem can be rewritten in augmented form
\begin{align}
&\min_{ \{ u_{1,k} \}_{k=0}^{N-1} } z_N^\top \bar Q_N z_N + \sum_{k=0}^{N-1}\!\big(z_k^\top \bar Q z_k + u_{1,k}^\top R_1 u_{1,k}\big), \label{eq:aug-cost}\\
& \text{subject to:} \quad z_{k+1}=\bar A_k z_k + \bar B_k u_{1,k}, \quad k =0,\dots, N-1,\label{eq:augAB}
\end{align}
with $\bar A_k,\bar B_k,\bar Q,\bar Q_N$ given in \eqref{eq:ABQ-def}.

Define the value function of the augmented system
\begin{equation}\label{eq:Vk}
V_k(z):=\min_{\{u_{1,\ell}\}_{\ell=k}^{N-1}}
\Big\{z_N^\top \bar Q_N z_N+\sum_{\ell=k}^{N-1}\!(z_\ell^\top \bar Q z_\ell + u_{1,\ell}^\top R_1 u_{1,\ell})\Big\}
\end{equation}
subject to the augmented dynamic constraint of \eqref{eq:augAB}.

Based on the Bellman principle of optimality at time $k$, we can reformulate $V_k(z)$ as
\begin{equation}
\label{eq:V_k vs V_k+1}
    V_k(z_k)=\min_{u_{1,k}}\Big\{ z_k^\top \bar{Q} z_k + u_{1,k}^\top R_1 u_{1,k} + V_{k+1}\big(z_{k+1}\big) \Big\},
\end{equation}

Backward induction yields $V_k(z_k)=z_k^\top \bar P_k z_k + \bar c_k$ with $\bar P_k$ given by \eqref{eq:riccati-aug}.
Substituting the augmented dynamics \eqref{eq:augAB} in \eqref{eq:V_k vs V_k+1} and collecting terms in $u_{1,k}$ yields the strictly convex quadratic in $u_{1,k}$:
\begin{equation}
\begin{aligned}
\mathcal Q_k(z_k,u_{1,k})
= u_{1,k}^\top &\mathcal H_k\, u_{1,k}
+ 2\,u_{1,k}^\top \bar B_k^\top \bar P_{k+1}\bar A_k\, z_k\\
&+ z_k^\top\!\big(\bar Q+\bar A_k^\top \bar P_{k+1}\bar A_k\big) z_k
+ \bar c_{k+1},
\end{aligned}
\end{equation}
where $\mathcal H_k:=R_1+\bar B_k^\top \bar P_{k+1}\bar B_k$ which is positive definite as $R_1\succ 0$ and $\bar P_{k+1}\succeq 0$.

Hence, the unique minimizer is
\begin{equation}\label{eq:ukstar}
u_{1,k} \;=\; -\mathcal H_k^{-1}\bar B_k^\top \bar P_{k+1}\bar A_k\, z_k =:-\big[\,K_k\ \ L_k\,\big] z_k
\end{equation}
which, upon partitioning $z_k=[x_k^\top \; w_k^\top]^\top$, gives \eqref{eq:gains-closed}. 

Completing the square yields \eqref{eq:riccati-aug}. Summing the stagewise completion terms from $k=0$ to $N-1$ telescopes the value functions and proves \eqref{eq:gap}. Uniqueness follows the strict convexity ($\mathcal H_k\succ 0$).
\end{proof}



By \eqref{eq:gap} in Theorem~\ref{thm:lag-optimal}, the gains $\{K_k, L_k\}$ are optimal over the causal affine class \eqref{eq:policy-class}, which includes the REF baseline $u_{1,k}^{\text{REF}}$ as the special case $L_k\equiv0$. 
Consequently, when execution-time deviations by Player 2 are measurable, Player~1 can improve upon the uncompensated reference feedback by adopting the compensated law within this policy class. This does not redefine the Nash equilibrium of the underlying game: the FNE is associated with the nominal game without implementation deviations, whereas the proposed law is a compensation strategy for the realized disturbed interaction.
The approach is practically relevant because such deviations routinely arise from disturbances, physical limits, and execution imperfections.

The next result quantifies the local benefit, giving a signed, second-order improvement in the deviation magnitude.

\begin{corollary}
\label{cor:local-gain-lag}
Let $J_1^{\mathrm{CF}}$ and $J_1^{\mathrm{REF}}$ be the costs under
$u_{1,k}^{\mathrm{CF}}=-K_k x_k - L_k w_k$ and
$u_{1,k}^{\mathrm{REF}}=-K^{\star}_{1,k} x_k$, respectively, where
$\{\bar P_k\}$ and $[\,K_k\ \ L_k\,]$ are as in Theorem~\ref{thm:lag-optimal}.
Define the selector $E_w:=[\,0\ \ I\,]$ so that $w_k=E_w z_k$ for
$z_k:=[x_k^\top\ w_k^\top]^\top$. For sufficiently small deviations
$\{\Delta u_{2,k}\}$ where $w_k=\mathcal{O}(\|\Delta u_2\|_2)$,
\begin{equation}\label{eq:local-gain-lag}
\begin{aligned}
J_1^{\mathrm{CF}} - J_1^{\mathrm{REF}}
\;=\;
-\,\sum_{k=0}^{N-1} \Big\| \mathcal{H}_k^{-\tfrac12}
\,\bar B_k^\top &\bar P_{k+1} E_w\, w_k \Big\|_2^{\,2}\\
&\;+\; \mathcal{O}\!\big(\|\Delta u_2\|_2^3\big)
\;\le\; 0,
\end{aligned}
\end{equation}
where $\displaystyle \mathcal{H}_k:=R_1+\bar B_k^\top \bar P_{k+1}\bar B_k \succ 0$. Thus, CF never performs worse than REF to second order in the deviation magnitude and strictly improves unless $w_k\equiv0$.
\end{corollary}
\begin{proof}
By Theorem~\ref{thm:lag-optimal}, the stagewise completion of squares gives, for any
candidate $\tilde u_{1,k}$,
$\mathcal{Q}_k(z_k,\tilde u_{1,k})-\mathcal{Q}_k(z_k,u_{1,k}^\star)
=\|\mathcal{H}_k^{1/2}(\tilde u_{1,k}-u_{1,k}^\star)\|_2^2$.
Evaluating this at the REF choice $\tilde u_{1,k}=-K_{1,k}^\star x_k$ yields
$\tilde u_{1,k}-u_{1,k}^\star= L_k w_k$ plus higher-order state terms induced by
$\Delta u_2$, which are $\mathcal{O}(\|\Delta u_2\|_2^2)$. Summing over $k$ gives
\eqref{eq:local-gain-lag} with the remainder $\mathcal{O}(\|\Delta u_2\|_2^3)$ and completes this proof.
\end{proof}



\subsection{Specialization: Slowly-Varying Error Dynamics}
When the FNE command $u_{2,k}^\star$ varies slowly, the drive term becomes negligible and $w_{k+1}=\alpha w_k$. The CF gains then admit closed forms in terms of the block Riccati sequence.

\begin{corollary}
\label{cor:slow-varying-alpha}
Assume the lag model \eqref{eq:w-lag} with $0<\alpha<1$ and that Player~2’s Nash command varies slowly so that
\begin{equation}\label{eq:slow-drive}
u^{\star}_{2,k+1}\approx u^{\star}_{2,k}
\end{equation}
hence the drive is negligible and $w_{k+1} = \alpha\,w_k$. Then the optimal gains $K_k$ and $L_k$ defined in Theorem~\ref{thm:lag-optimal} will be obtained as
\begin{equation}
\begin{aligned}
&K_k \;=\; H_k^{-1} B_1^\top P_{k+1}^{xx} A_{\mathrm{eff},k}, \\
&L_k \;=\; H_k^{-1} B_1^\top \big(P_{k+1}^{xx} + \alpha\,P_{k+1}^{xw}\big), \label{eq:special case k}
\end{aligned}
\end{equation}
where $H_k \;:=\; R_1 + B_1^\top P_{k+1}^{xx} B_1 \ \succ 0$; $P_{k}^{xx}$ satisfies the standard finite-horizon Riccati recursion
\begin{equation}\label{eq:Riccati-Pxx}
\begin{aligned}
&P_N^{xx}=Q_{1,N},\\
&\begin{aligned}
P_k^{xx}=Q_1 + &A_{\mathrm{eff},k}^\top P_{k+1}^{xx} A_{\mathrm{eff},k}\\
&- A_{\mathrm{eff},k}^\top P_{k+1}^{xx} B_1 H_k^{-1} B_1^\top P_{k+1}^{xx} A_{\mathrm{eff},k},
\end{aligned}
\end{aligned}
\end{equation}
and the cross block $P_k^{xw}$ evolves linearly as
\begin{equation}\label{eq:Riccati-Pxw}
\begin{aligned}
&P_N^{xw}=0,\\
&\begin{aligned}
P_k^{xw} =\; &A_{\mathrm{eff},k}^\top P_{k+1}^{xx}
+ \alpha\,A_{\mathrm{eff},k}^\top P_{k+1}^{xw}\\
&- A_{\mathrm{eff},k}^\top P_{k+1}^{xx} B_1 H_k^{-1} B_1^\top \big(P_{k+1}^{xx}+\alpha\,P_{k+1}^{xw}\big).
\end{aligned}
\end{aligned}
\end{equation}
\end{corollary}

\begin{proof}
The assumption \eqref{eq:slow-drive} implies $w_{k+1}=\alpha w_k$. In augmented coordinates
$z_k=[x_k^\top \; w_k^\top]^\top$ this gives
\[
\bar A_k=\begin{bmatrix}A_{\mathrm{eff},k}& I\\ 0 & \alpha I\end{bmatrix},\qquad
\bar B_k=\begin{bmatrix}B_1\\ 0\end{bmatrix}.
\]
By Theorem~\ref{thm:lag-optimal}, the optimal control \eqref{eq:ukstar} can be simplified by partition $\bar P_{k+1}=\begin{bmatrix}P_{k+1}^{xx}&P_{k+1}^{xw}\\ P_{k+1}^{wx}&P_{k+1}^{ww}\end{bmatrix}$ and computing the terms
\begin{equation*}
\bar B_k^\top \bar P_{k+1}\bar B_k \;=\; B_1^\top P_{k+1}^{xx} B_1,
\end{equation*}
\begin{equation*}
\bar B_k^\top \bar P_{k+1}\bar A_k
= B_1^\top\!\big[\,P_{k+1}^{xx}A_{\mathrm{eff},k}\ \ \ P_{k+1}^{xx}+\alpha P_{k+1}^{xw}\,\big].
\end{equation*}
Defining $H_k := R_1 + B_1^\top P_{k+1}^{xx} B_1 \succ 0$, 
\begin{equation}
\begin{aligned}
u_{1,k}^\star
&= -\,H_k^{-1}\,B_1^\top\!\big[\,P_{k+1}^{xx}A_{\mathrm{eff},k}\ \ \ P_{k+1}^{xx}+\alpha P_{k+1}^{xw}\,\big]
\begin{bmatrix}x_k\\ w_k\end{bmatrix}\\[2pt]
&\;\begin{aligned}
= &-\,\underbrace{H_k^{-1}B_1^\top P_{k+1}^{xx}A_{\mathrm{eff},k}}_{\displaystyle K_k}\,x_k\\
&\;-\;\underbrace{H_k^{-1}B_1^\top\big(P_{k+1}^{xx}+\alpha P_{k+1}^{xw}\big)}_{\displaystyle L_k}\,w_k,
\end{aligned}
\end{aligned}
\end{equation}
which yields \eqref{eq:special case k}. Taking the $(1,1)$ and $(1,2)$ blocks of the Riccati recursion
\eqref{eq:riccati-aug} with the above $\bar A_k,\bar B_k$ gives the stated recursions
\eqref{eq:Riccati-Pxx}–\eqref{eq:Riccati-Pxw} and $P_N^{xw}=0$ follows from $\bar Q_N$ being block-diagonal.
\end{proof}

Comparing the Riccati recursion $P^{xx}_k$ in \eqref{eq:Riccati-Pxx}, and the gain formulation $K_k$ in \eqref{eq:special case k} with the ones in Definition~\ref{def: definition of Nash} reveals that the state gain $K_k$ obtained under the assumption of Corollary~\ref{cor:slow-varying-alpha} is equal to that of the FNE.

It is worth noting that \eqref{eq:Riccati-Pxx} coincides with the standard finite-horizon Riccati recursion for the disturbance-free system. In particular, $K_k$ matches the usual state-feedback gain derived from $P_k^{xx}$, while $L_k$ adds a causal feedforward term shaped by both $P_{k+1}^{xx}$ and the cross block $P_{k+1}^{xw}$ scaled by $\alpha$. The resulting law $u_{1,k}^\star=-K_k x_k - L_k w_k$ thus has the familiar structure of disturbance-accommodation or preview control: the controller reacts not only to the state but also to the measured disturbance. Under the lag generator \eqref{eq:w-lag}, the augmented Riccati \eqref{eq:riccati-aug} introduces off-diagonal couplings between $x$ and $w$; these generate a preview component in $L_k$ that anticipates how current deviations propagate through the lag into future stages. Hence \eqref{eq:gains-closed} recovers the classical disturbance-accommodation principle and extends it to physically grounded, time-varying deviation dynamics.

\begin{remark}
To sanity-check the design, we examine the two extreme points of the lag parameter 
$\alpha$ in the implemented-command model \eqref{eq:u2-lag}. 
When $\alpha\to 0$, the actuator tracks its Nash command essentially instantaneously, so 
$u_{2,k+1}\approx u^{\star}_{2,k}$ and the disturbance recursion reduces to 
$w_{k+1}=-B_2(u^{\star}_{2,k+1}-u^{\star}_{2,k})$. In steady segments the disturbance vanishes, 
and the compensated feedback coincides with the classical disturbance-accommodation (DBSC) law 
that uses $w_k$ as an immediate feedforward signal. At the other extreme, when $\alpha\to 1$ and 
the drive term is negligible, the disturbance behaves like a nearly constant bias with 
$w_{k+1}\approx w_k$, so the compensation acts as preview cancellation of a persistent input. 
Both cases are contained as limits of the proposed lag model and show that the compensated law 
weakly dominates any $K$-only feedback whenever a measurable deviation is present.
\end{remark}

\begin{remark}
\label{rem:CF-homog}
Under the CF policy for Player~1 \eqref{eq:policy-class} and the feedback policy of Player~2 \eqref{eq:feedback nash definition}, the closed-loop state update can be written as
\[
x_{k+1}=\big(A-B_1K_{1,k}-B_2K^\star_{2,k}\big)\,x_k - B_1L_k w_k.
\]
Hence, even for $w_k\equiv 0$, the homogeneous (zero-disturbance) dynamics are
\[
x_{k+1}=\big({A-B_1K_{1,k}-B_2K^\star_{2,k}}\big)\,x_k,
\]
which generally differ from the FNE dynamics
$x_{k+1}=\big(A-B_1K^\star_{1,k}-B_2K^\star_{2,k}\big)\,x_k$.
Stability and performance must therefore be assessed for the Linear Time-Varying (LTV) matrix sequence produced by the augmented Riccati recursion.
\end{remark}

It is noteworthy that the CF policy is not a Nash equilibrium for the modified interaction, since Player~2 is not best-responding to CF. Joint redesign of both players under implementation dynamics is more involved and is left for future work.

The next section presents a numerical example illustrating the practical utility of Theorems~\ref{thm:lag-optimal}.

\section{Numerical Example: Two Carts Coupled by a Spring--Damper}
\label{sec:numerics}

We consider a two-player, finite-horizon, discrete-time LQ dynamic game with feedback Nash equilibrium (FNE) strategies. Two equal-mass carts are connected by a spring--damper; Player~1 and Player~2 actuate Cart~1 and Cart~2, respectively. Each player regulates its cart to the origin while penalizing control effort. We then impose a realistic execution-lag deviation at Player~2 and show that Player~1 can employ the CF from Theorem~\ref{thm:lag-optimal}, computed via an augmented Riccati recursion, to counteract that structured disturbance while Player~2 remains at FNE.
\begin{figure}
    \centering
    \includegraphics[width=1\linewidth]{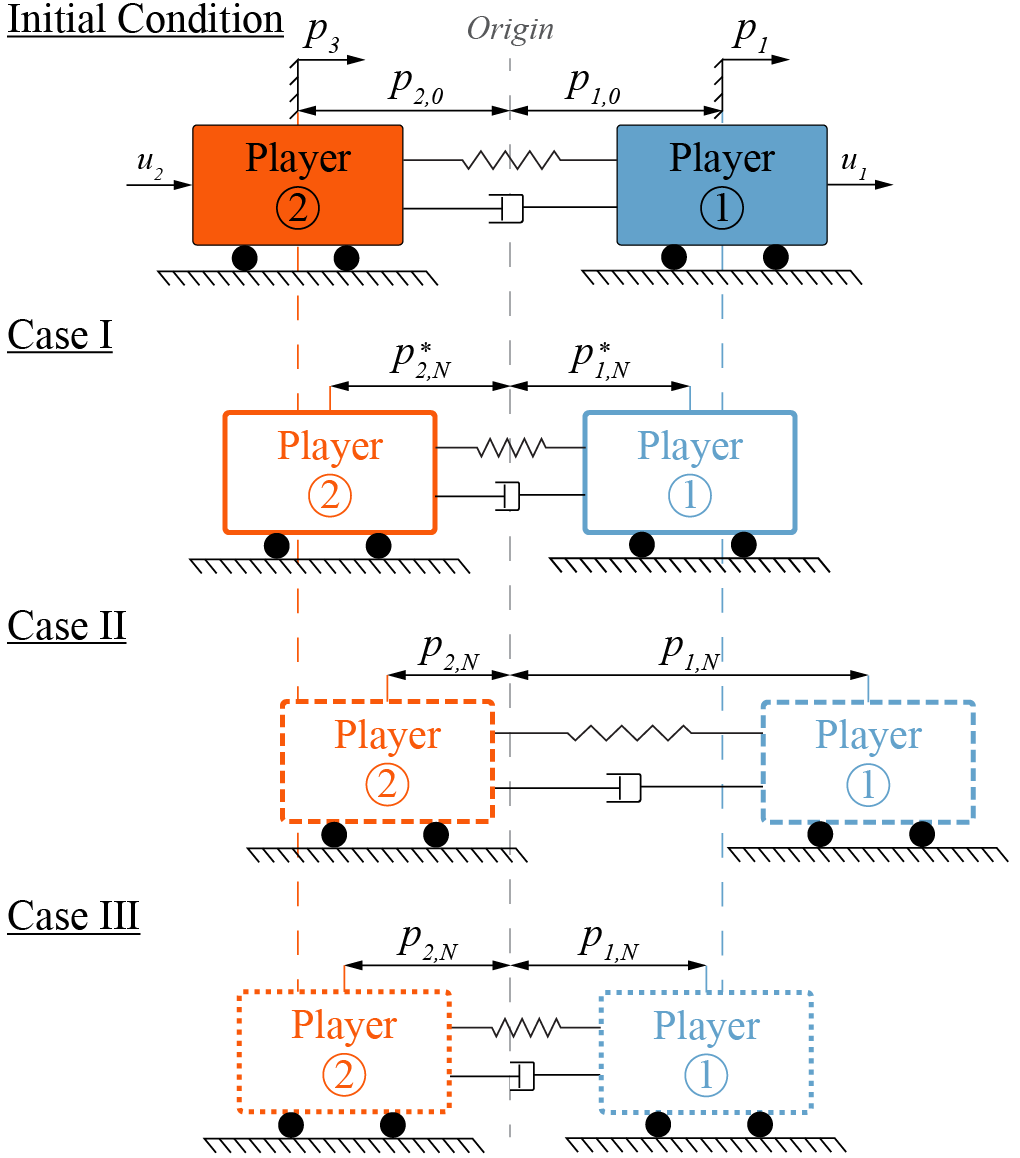}
    \caption{Schematic of the two-cart setup and qualitative final positions under the three cases. 
Case~I (FNE) shows near-symmetric convergence near the origin. 
Case~II (REF) illustrates how Player~2’s actuator lag degrades both players’ positions relative to nominal.  
Case~III (CF) shows that the compensated policy enables Player~1 to mitigate the error and approach its nominal outcome, while Player~2 remains misaligned due to its uncompensated lag.}
    \label{fig:system}
\end{figure}
Let $x:=\begin{bmatrix}p_1 & v_1 & p_2 & v_2\end{bmatrix}^\top$, where $p_1$ and $p_2$ are the cart positions, and $v_1$ and $v_2$ are the corresponding  velocities, and $u_1$ and $u_2$ denote the longitudinal actuation forces applied by Player~1 and~2, respectively. 
The continuous-time system dynamics are $\dot x = A_c x + B_{1,c} u_1 + B_{2,c} u_2$ with the matrices:
\begin{equation*}
\begin{aligned}
&A_c = \begin{bmatrix}
0&1&0&0\\[2pt]
-\frac{k}{m}&-\frac{c}{m}&\frac{k}{m}&\frac{c}{m}\\[2pt]
0&0&0&1\\[2pt]
\frac{k}{m}&\frac{c}{m}&-\frac{k}{m}&-\frac{c}{m}
\end{bmatrix},\\
&B_{1,c}=\begin{bmatrix}0\\ \frac{1}{m}\\ 0\\ 0\end{bmatrix},\quad
B_{2,c}=\begin{bmatrix}0\\ 0\\ 0\\ \frac{1}{m}\end{bmatrix}.
\end{aligned}
\end{equation*}
We discretize the dynamics via forward Euler at sampling $\Delta T$. Each player minimizes the quadratic cost in \eqref{eq:LQ cost function} over a horizon $N$ with the weights $Q_1=\mathrm{diag}(10,\,2,\,0,\,0)$, $Q_2=\mathrm{diag}(0,\,0,\,10,\,2)$, and $R_1=R_2=[\,0.05\,]$.

For error identification purposes, let $u_{2,k}^\star=-K_{2,k}^\star x_k$ denote Player~2’s intended (FNE) command, and let $u_{2,k}$ be the applied input after a first-order discrete lag \eqref{eq:u2-lag} with the $\alpha$ defined in Remark~\ref{rem:alpha-tau}. The parameters used are listed in Table~\ref{tab:params}.

\begin{table}[ht]
\centering
\caption{Parameters used in the numerical simulations.}
\begin{tabular}{lcl}
\hline
\multicolumn{3}{c}{\textbf{Model parameters}} \\
\hline
$m$       & $1.0$              & Cart mass [kg] \\
$k$       & $0.12$             & Spring constant [N/m] \\
$c$       & $30$               & Damping coefficient [N$\cdot$s/m] \\
$\Delta T$& $0.015$            & Sampling period [s] \\
$N$       & $40$               & Horizon length [steps] \\
\hline
\multicolumn{3}{c}{\textbf{Error/lag parameters}} \\
\hline
$\tau$    & $0.8$              & Lag time constant [s] \\
$\alpha$  & $e^{-\Delta T/\tau}\approx 0.9814$ & Discrete lag factor [–] \\
\hline
\multicolumn{3}{c}{\textbf{Initial conditions}} \\
\hline
$x_0$     & $\begin{bmatrix}0.5 & 0 & -0.5 & 0\end{bmatrix}^\top$
          & Initial state [m, m/s, m, m/s] \\
\hline
\end{tabular}
\label{tab:params}
\end{table}

\subsection*{Test Cases}
We compare three implementations over the same horizon and initial condition.
\emph{\textbf{Case~I (FNE)}:} both players apply their finite-horizon feedback Nash laws computed from the coupled Riccati recursion; no actuator lag is present, and Player~2’s intended command is applied directly.
\emph{\textbf{Case~II (REF)}:} Player~2’s intended command is subjected to a first-order discrete lag with factor $\alpha$, but Player~1 continues to use its nominal feedback policy; this isolates the performance loss due to the execution lag.
\emph{\textbf{Case~III (CF)}:} the same actuator lag is present on Player~2, while Player~1 designs a compensated feedback by solving a finite-horizon Riccati recursion on the exact augmented realization of the lagged plant, yielding time-varying gains that act on the state and on the lag-induced mismatch signal. Figure~\ref{fig:system} summarizes the two-cart configuration and qualitatively contrasts the resulting final positions under the three cases.

Figure~\ref{fig:state-input comparison} compares the state and input trajectories under the three scenarios. 
In the nominal FNE (Case~I, solid lines), both players regulate their carts symmetrically to the origin with smooth, coordinated actions. 
With a first-order execution lag at Player~2 but no compensation (Case~II, dashed), the applied input $u_2$ deviates from the intended FNE signal (step-like due to the zero-order hold), producing a sluggish response in $(p_2,v_2)$ and, through the spring--damper coupling, a degraded transient in $(p_1,v_1)$. 
This mismatch yields a significant cost increase for both players (Table~\ref{tab:costs}).

\begin{figure}
    \centering
    \includegraphics[width=0.95\linewidth]{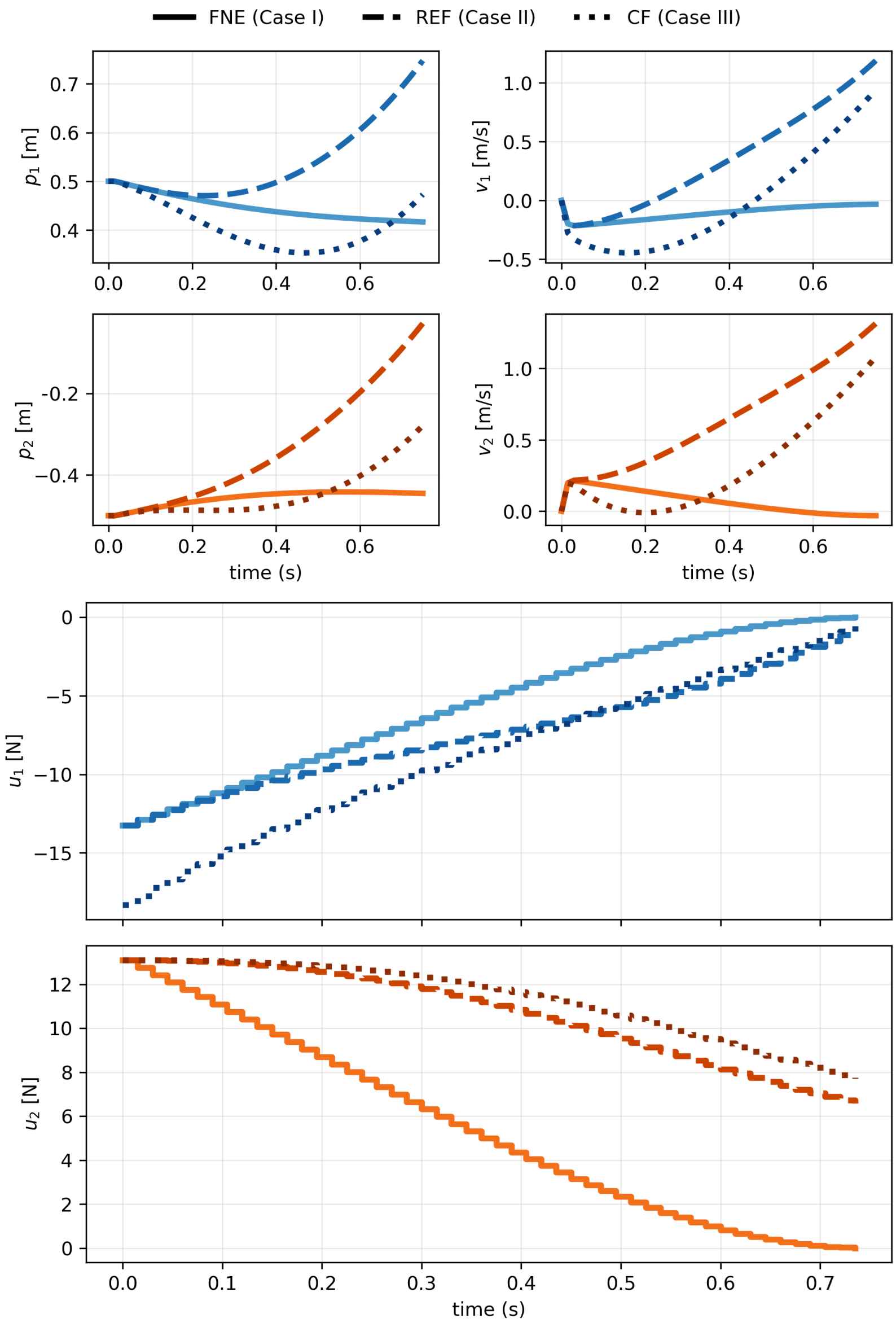}
    \caption{State trajectories $(p_1,v_1,p_2,v_2)$ and applied inputs under the three cases. The compensated feedback (CF) attenuates the effect of Player~2’s actuator lag by shaping a mismatch-aware control for Player~1.}
    \label{fig:state-input comparison}
\end{figure}

Under the compensated feedback (Case~III, dotted), Player~1 modifies its input relative to the nominal FNE action to proactively counteract the actuation shortfall of Player~2. 
The resulting trajectories stay appreciably closer to the FNE benchmark, especially for Player~1, confirming that CF recovers a substantial fraction of the lag-induced performance loss. 
Player~2’s cost, however, rises further because it neither alters its strategy nor removes the actuator lag, highlighting the asymmetric benefit in a general-sum game.

Numerically, Table~\ref{tab:costs} shows that Player~1's cost $J_1$ increases by $51.36\%$ from Case~I to Case~II, when lag is present without compensation. Under CF in Case~III, Player~1's cost decreases to $439.67$, yielding a $10.71\%$ improvement relative to the uncompensated lagged case and recovering about $31.56\%$ of the nominal-to-lag performance loss. At the same time, Player~2's cost $J_2$ increases from $492.11$ in Case~II to $556.40$ in Case~III.

\begin{table}[h]
\caption{Finite-horizon performance under the three cases ($\tau=0.8$\,s, $\alpha\approx0.9814$). CF (Case~III) reduces Player~1’s lag penalty
relative to Case~II.}
\centering
\begin{tabular}{lccc}
\hline
 & Case I: FNE & Case II: Lag, no comp. & Case III: CF \\
\hline
$J_1$        & 325.31   & 492.40   & \textbf{439.67} \\
$J_2$        & 330.45   & 492.11   & 556.40 \\
\hline
$\mathrm{ISE}_1$ & 0.1655   & 0.4566   & {0.2681} \\
$\mathrm{ISE}_2$ & 0.1689   & 0.5076   & {0.3113} \\
\hline
\end{tabular}
\label{tab:costs}
\end{table}

To further assess transient regulation quality, Table~\ref{tab:costs} also reports the finite-horizon integral squared error (ISE) for each player, where $\mathrm{ISE}_1$ corresponds to Player~1 and $\mathrm{ISE}_2$ corresponds to Player~2. These values show the same qualitative improvement for Player~1: relative to Case~II, $\mathrm{ISE}_1$ is reduced by $41.28\%$ under CF. Interestingly, Player~2's $\mathrm{ISE}_2$ also decreases, by $38.67\%$, even though Player~2's cost increases. This indicates that the increase in $J_2$ is not due solely to poorer regulation. Rather, because Player~2's objective penalizes both state deviation and control effort, the modified interaction can reduce regulation error while still increasing the overall quadratic cost.

However, the total cost is increased. This tradeoff is expected in the present formulation. The proposed CF law is obtained by minimizing Player~1's cost over the causal affine class~\eqref{eq:policy-class} while Player~2 remains fixed at its nominal FNE strategy and still suffers from actuator lag. Therefore, the design is intentionally player-specific rather than socially optimal, and a reduction in $J_1$ does not in general imply a reduction in $J_1+J_2$. Mitigating this efficiency tradeoff would require a joint redesign of both players under implementation dynamics, or alternatively a cooperative/team objective, which lies beyond the scope of the present work.



\section{Conclusion}
\label{sec: conclusion}
In this paper, we analyzed the impact of control implementation errors in finite-horizon, two-player, nonzero-sum LQ games. 
We quantified how deviations from one player’s nominal
feedback Nash equilibrium input perturb the closed-loop
trajectory and affect the other player’s cost. In particular,
we derived explicit first-order sensitivity expressions that
characterize this dependence in the discrete-time setting.

Motivated by this sensitivity analysis, we constructed a compensated feedback law for the influenced player by augmenting the nominal dynamics with measurable deviation dynamics. The resulting controller is optimal over a causal affine policy class and includes the uncompensated equilibrium-derived
feedback as a special case.

A representative numerical example illustrated how the proposed
compensation can recover a substantial portion of the performance
loss induced by actuator lag. These results highlight the value
of exploiting measurable implementation deviations when the
exact execution of nominal Nash policies is not possible.

\subsection*{Future Research Directions}
A promising avenue for future research is the development of online estimation techniques combined with adaptive control laws that can anticipate deviations of the other player from the Nash strategy and dynamically adjust to these errors in real time. Furthermore, examining the effects of more complex deviation patterns, such as time-varying, correlated, or structured errors, could offer deeper insights into the robustness and adaptability of game-theoretic control strategies. 

\section{Acknowledgment}
The authors acknowledge the use of OpenAI’s ChatGPT for structuring content, refining language, and improving clarity. All technical contributions, derivations, and original insights remain the authors' own work.

\bibliographystyle{ieeetr}
\bibliography{references}

@book{basar1999dynamic,
  author = {T. Basar and G. J. Olsder},
  title = {Dynamic Noncooperative Game Theory},
  publisher = {SIAM},
  year = {1999}
}

@article{nash1950equilibrium,
  author = {J. F. Nash},
  title = {Equilibrium points in n-person games},
  journal = {Proceedings of the National Academy of Sciences},
  volume = {36},
  number = {1},
  pages = {48--49},
  year = {1950}
}

@article{shamma2020game,
author = {Shamma, Jeff},
year = {2020},
month = {07},
pages = {1118-1119},
title = {Game theory, learning, and control systems},
volume = {7},
journal = {National Science Review},
doi = {10.1093/nsr/nwz163}
}

@article{LAA_Jungers_2013,
  author    = {Marc Jungers},
  title     = {Feedback Strategies for Discrete-Time Linear-Quadratic Two-Player Descriptor Games},
  journal   = {Linear Algebra and its Applications},
  volume    = {439},
  number    = {10},
  pages     = {3111--3134},
  year      = {2013},
  publisher = {Elsevier},
  doi       = {10.1016/j.laa.2013.07.015},
  url       = {https://doi.org/10.1016/j.laa.2013.07.015}
}

@article{NORTMANN20231760,
title = {Approximate Nash Equilibria for Discrete-Time Linear Quadratic Dynamic Games},
journal = {IFAC-PapersOnLine},
volume = {56},
number = {2},
pages = {1760-1765},
year = {2023},
note = {22nd IFAC World Congress},
issn = {2405-8963},
doi = {https://doi.org/10.1016/j.ifacol.2023.10.1886},
url = {https://www.sciencedirect.com/science/article/pii/S2405896323022954},
author = {Benita Nortmann and Thulasi Mylvaganam}
}

@article{AMATO2002507_guaranteeing,
author = {Amato, F. and Mattei, M. and Pironti, A.},
title = {Guaranteeing cost strategies for linear quadratic differential games under uncertain dynamics},
journal = {Automatica},
volume = {38},
number = {3},
pages = {507-515},
year = {2002},
issn = {0005-1098},
doi = {https://doi.org/10.1016/S0005-1098(01)00224-2},
url = {https://www.sciencedirect.com/science/article/pii/S0005109801002242}
}

@article{Jimenez01072006,
author = {Manuel Jimenez and Alex Poznyak},
title = {$\varepsilon$-Equilibrium in LQ Differential Games with Bounded Uncertain Disturbances: Robustness of Standard Strategies and New Strategies with Adaptation},
journal = {International Journal of Control},
volume = {79},
number = {7},
pages = {786--797},
year = {2006},
publisher = {Taylor \& Francis},
doi = {10.1080/00207170600690624},
URL = {https://doi.org/10.1080/00207170600690624},
eprint = {https://doi.org/10.1080/00207170600690624}
}

@Article{van-den=broek2023,
  author={W. A. van den Broek and J. C. Engwerda and J. M. Schumacher},
  title={{Robust Equilibria in Indefinite Linear-Quadratic Differential Games}},
  journal={Journal of Optimization Theory and Applications},
  year={2003},
  volume={119},
  number={3},
  pages={565-595},
  month={December},
  doi={10.1023/B:JOTA.0000006690},
  url={https://ideas.repec.org/a/spr/joptap/v119y2003i3d10.1023_bjota.0000006690.78564.88.html}
}

@article{TANAKA1991413,
title = {On $\varepsilon$-equilibrium point in a noncooperative n-person game},
journal = {Journal of Mathematical Analysis and Applications},
volume = {160},
number = {2},
pages = {413-423},
year = {1991},
issn = {0022-247X},
doi = {https://doi.org/10.1016/0022-247X(91)90314-P},
url = {https://www.sciencedirect.com/science/article/pii/0022247X9190314P},
author = {Kensuke Tanaka and Kazunori Yokoyama}
}

@inproceedings{tan2005sensitivity,
  title={Sensitivity Analysis of Uncertainties in Dynamic Noncooperative Games Using a Multi-Modeling Method},
  author={Tan, Xiaohuan and Cruz, JB},
  booktitle={Proceedings of the 44th IEEE Conference on Decision and Control},
  pages={5113--5118},
  year={2005},
  organization={IEEE}
}

@inproceedings{chiu2021encoding,
  title={Encoding defensive driving as a dynamic Nash game},
  author={Chiu, Chih-Yuan and Fridovich-Keil, David and Tomlin, Claire J},
  booktitle={2021 IEEE International Conference on Robotics and Automation (ICRA)},
  pages={10749--10756},
  year={2021},
  organization={IEEE}
}

@Article{Oszczypała2023Nash,
AUTHOR = {Oszczypała, Mateusz and Ziółkowski, Jarosław and Małachowski, Jerzy and Lęgas, Aleksandra},
TITLE = {Nash Equilibrium and Stackelberg Approach for Traffic Flow Optimization in Road Transportation Networks—A Case Study of Warsaw},
JOURNAL = {Applied Sciences},
VOLUME = {13},
YEAR = {2023},
NUMBER = {5},
ARTICLE-NUMBER = {3085},
URL = {https://www.mdpi.com/2076-3417/13/5/3085},
ISSN = {2076-3417},
DOI = {10.3390/app13053085}
}

@article{YANG2020216nashQ,
title = {Nash Q-learning based equilibrium transfer for integrated energy management game with We-Energy},
journal = {Neurocomputing},
volume = {396},
pages = {216-223},
year = {2020},
issn = {0925-2312},
doi = {https://doi.org/10.1016/j.neucom.2019.01.109},
url = {https://www.sciencedirect.com/science/article/pii/S0925231219304357},
author = {Lingxiao Yang and Qiuye Sun and Dazhong Ma and Qinglai Wei}
}

@Article{Zheng2023game,
AUTHOR = {Zheng, Weimin and Meng, Fanying and Liu, Ning and Huang, Shuo},
TITLE = {A Game Model for Analyzing Wireless Sensor Networks of 5G Environment Based on Adaptive Equilibrium Optimizer Algorithm},
JOURNAL = {Sensors},
VOLUME = {23},
YEAR = {2023},
NUMBER = {19},
ARTICLE-NUMBER = {8055},
URL = {https://www.mdpi.com/1424-8220/23/19/8055},
PubMedID = {37836885},
ISSN = {1424-8220},
DOI = {10.3390/s23198055}
}

@article{Engwerda2009lq,
author = {Engwerda, Jacob},
year = {2009},
month = {03},
pages = {37-71},
title = {Linear Quadratic Differential Games: An Overview},
volume = {10},
isbn = {978-0-8176-4833-6},
journal = {Advances in Dynamic Games and their Applications},
doi = {10.1007/978-0-8176-4834-3_3}
}

@article{alvarez2009urban,
  title={Urban traffic control via stackelber-nash equilibria},
  author={Alvarez, Israel and Alexander, Villalobos and Poznyak, S},
  journal={IFAC Proceedings Volumes},
  volume={42},
  number={15},
  pages={582--587},
  year={2009},
  publisher={Elsevier}
}

@book{rugh1996linear,
  title={Linear system theory},
  author={Rugh, Wilson J},
  year={1996},
  publisher={Prentice-Hall, Inc.}
}

@article{abdolmohammadi2024sizing,
  title={Sizing and Life Cycle Assessment of Small-Scale Power Backup Solutions: A Statistical Approach},
  author={Abdolmohammadi, Armin and Nemati, Alireza and Haas, Meridian and Nazari, Shima},
  journal={IEEE Access},
  year={2024},
  publisher={IEEE}
}

@article{nazari2024powertrain,
  title={Powertrain Hybridization for Autonomous Vehicles: Fuel Efficiency Perspective in Mixed Autonomy Traffic},
  author={Nazari, Shima and Gowans, Norma and Abtahi, Mohammad and Rabbani, Mahdis},
  journal={IEEE Transactions on Transportation Electrification},
  year={2024},
  publisher={IEEE}
}

@article{abtahi2025multi,
  title={Multi-Step Deep Koopman Network (MDK-Net) for Vehicle Control in Frenet Frame},
  author={Abtahi, Mohammad and Rabbani, Mahdis and Abdolmohammadi, Armin and Nazari, Shima},
  journal={arXiv preprint arXiv:2503.03002},
  year={2025}
}

@article{bouzidi2025interaction,
  title={Interaction-aware motion planning and control based on game theory with human-in-the-loop validation},
  author={Bouzidi, Mohamed-Khalil and Hashemi, Ehsan},
  journal={Robotics and Autonomous Systems},
  volume={186},
  pages={104908},
  year={2025},
  publisher={Elsevier}
}

@article{hernandez2024ϵ,
  title={$\varepsilon$-Nash equilibrium of non-cooperative Lagrangian dynamic games based on the average sub-gradient robust integral sliding mode control},
  author={Hernandez Sanchez, Alejandra and Poznyak, Alexander and Chairez, Isaac},
  journal={International Journal of Robust and Nonlinear Control},
  volume={34},
  number={15},
  pages={10365--10387},
  year={2024},
  publisher={Wiley Online Library}
}

@article{salizzoni2025bridging,
  title={Bridging Finite and Infinite-Horizon Nash Equilibria in Linear Quadratic Games},
  author={Salizzoni, Giulio and Hall, Sophie and Kamgarpour, Maryam},
  journal={arXiv preprint arXiv:2508.20675},
  year={2025}
}

@article{soltanian2025peerawarecostestimationnonlinear,
  title={Peer-Aware Cost Estimation in Nonlinear General-Sum Dynamic Games for Mutual Learning and Intent Inference},
  author={Soltanian, Seyed Yousef and Zhang, Wenlong},
  journal={arXiv preprint arXiv:2504.17129},
  year={2025}
}

@article{ma2024ϵ,
  title={$\varepsilon$-Nash mean-field games for stochastic linear-quadratic systems with delay and applications},
  author={Ma, Heping and Shi, Yu and Li, Ruijing and Wang, Weifeng},
  journal={Probability, Uncertainty and Quantitative Risk},
  volume={9},
  number={3},
  pages={389--404},
  year={2024},
  publisher={Probability, Uncertainty and Quantitative Risk}
}

@article{guerrero2021openloop,
  title     = {Open-Loop Robust {N}ash Strategies in Discrete-Time Uncertain Dynamic Games},
  author    = {Guerrero, Jose and Delgado, Pablo and Martinez, Lucia},
  journal   = {IEEE Transactions on Automatic Control},
  volume    = {66},
  number    = {8},
  pages     = {2043--2058},
  year      = {2021},
  publisher = {IEEE}
}

@inproceedings{wang2022risk,
  title     = {A Risk-Sensitive Dynamic Game Approach for Multi-Agent Interactions Under Uncertainty},
  author    = {Wang, Xi and Chen, Yu and Dou, Wei and Li, Fang},
  booktitle = {Proceedings of the IEEE International Conference on Robotics and Automation (ICRA)},
  year      = {2022},
  pages     = {13089--13096},
  organization = {IEEE}
}

\end{document}